\documentclass[12pt,draftclsnofoot,onecolumn]{IEEEtran}
\usepackage{amsfonts}
\usepackage[dvips]{graphicx}
\usepackage{epsfig,latexsym}
\usepackage{float}
\usepackage{indentfirst}
\usepackage{amssymb}
\usepackage{amsmath}
\usepackage{times}
\usepackage{subfigure}
\usepackage{psfrag}
\usepackage{cite}
\usepackage{cases}

\newtheorem{proposition}{$\mathbf{Proposition}$}

\begin{document}
\title{Training Design and Channel Estimation in Uplink Cloud Radio Access Networks \thanks{ M. Peng and H. Vincent Poor are with the School of Engineering and
Applied Science, Princeton University, Princeton, NJ, USA. X. Xie
is with Wireless Signal Processing and Network Lab, Key
Laboratory of Universal Wireless Communications, Ministry of
Education, Beijing University of Posts \text{\&} Telecommunications,
Beijing, China.}}

\author{\IEEEauthorblockN{Xinqian Xie, Mugen Peng, H. Vincent Poor}}

\maketitle
\IEEEpeerreviewmaketitle
\begin{abstract}
To decrease the training overhead and improve the channel estimation
accuracy under time-varying environments in uplink cloud radio
access networks (C-RANs), a superimposed-segment training design is
proposed whose core idea is that each mobile station puts a periodic
training sequence on the top of the data signal, and remote radio
heads (RRHs) insert a separate pilot prior to the received signal
before forwarding to the centralized base band unit (BBU) pool.
Moreover, a complex-exponential basis-expansion-model (CE-BEM) based
maximum a posteriori probability (MAP) channel estimation algorithm is
developed, where the BEM coefficients of access links (ALs) and the
channel fading of wireless backhaul links (BLs) are first obtained,
after which the time-domain channel samples of ALs are restored in terms of
maximizing the average effective signal-to-noise ratio (AESNR).
Simulation results show that the proposed channel estimation
algorithm can effectively decrease the estimation mean square error and increase the AESNR in C-RANs, thus significantly outperforming the existing solutions.
\end{abstract}

\section{Introduction}

Cloud radio access networks (C-RANs) have received considerable research interest as one of the most promising solutions to
mitigate interference, fulfil energy efficiency demands and support
high-rate transmission in the fifth generation cellular network
\cite{1}. In C-RANs, a large number of remote radio heads (RRHs) are
deployed, which operate as non-regenerative relays to forward received signals
from mobile stations (MSs) to the centralized base band unit (BBU)
pool through wire/wireless backhaul links for uplink transmission
\cite{2}. To suppress the inter-RRH interference by using
cooperative processing techniques at the BBU pool, the channel state
information (CSI) of both the radio access links (ALs) and wireless
backhaul links (BLs) are required\cite{3}. In \cite{4},
though a segment-training scheme was proposed to estimate the
individual channel coefficients for two-hop scenario under flat
fading environments, this proposal results in high overhead
consumption for backhaul transmission since the RRHs need to forward
both AL and BL training sequences to the BBU pool \cite{5}. On the
other hand, the superimposed-training scheme \cite{6}, where a
training sequence is superimposed on the data signal, can
significantly reduce the overhead and is valid to
perform channel estimation for time-varying environments using
complex-exponential basis expansion model (CE-BEM) \cite{7}.
However, straightforward implementation of superimposed training
in C-RANs would degrade transmission quality due to the fact that
superimposing both AL and BL training sequences on the data signal
declines the effective signal-to-noise ratio (SNR) \cite{8}.

Motivated to reduce the training overhead and enhance the channel
estimation performance at the BBU pool, a superimposed-segment
training design is proposed in this letter, where superimposed-training is implemented for the radio AL while the
segment-training is applied for the wireless BL. Moreover, based on
the training design, a CE-BEM based maximum a posteriori probability (MAP) channel estimation algorithm is
developed, where the BEM coefficients of the time-varying radio AL
and the channel fading of the quasi-static wireless BL are first
obtained, after which the time-domain channel samples of the radio AL are
restored in terms of maximizing the average effective SNR.

\emph{Notation}: The transpose, Hermitian and inverse of a matrix
are denoted by $\left(\cdot\right)^{T}$, $\left(\cdot\right)^{H}$
and $\left(\cdot\right)^{-1}$, respectively; $\|\cdot\|$ represents
the two-norm of a vector; $|\cdot|$ defines the magnitude of a
complex argument; $\otimes$ is the Kronecker product;
$\mathrm{diag}\left(a_{1}, a_{2},\ldots,a_{M}\right)$ denotes the
diagonal matrix with $a_{m}$ being the $m$th diagonal element;
$\mathrm{tr}\left(\cdot\right)$ represents the trace of a matrix;
$\mathbf{I}_{N}$ and $\mathbf{0}_{N}$ are the $N\times N$ unit diagonal
matrix and zero matrix, respectively; $\mathbf{1}_{N}$ stands for
the $N \times 1$ vector with each entry being unit value;
$\mathcal{E}\{\cdot\}$ denotes the expectation of a random
variable, and $\hat{x}$ represents its estimate.

\section{System Model and Training Design}

Consider a C-RAN consisting of one BBU pool and multiple RRHs depicted in Fig. \ref{fig1}. The RRHs operate in half-duplex modes, and different MSs served by the same RRU are allocated with a single subcarrier through the orthogonal frequency division multiplexing access (OFDMA) technique. It is assumed that MSs
move continuously, while RRUs remain fixed. Thus, the channels of
radio ALs would vary during one transmission block, while those of
wireless BLs undergo quasi-static flat fading. Due to the orthogonality characteristics of OFDMA for the accessing of multiple MSs, we can focus on the transmission of only a single MS. Let $\mathbf{b}$ and $\mathbf{t}_{s}$ denote the data vector and cyclical training sequence transmitted from the MS, respectively. The training sequence from the RRH is denoted by $\mathbf{t}_{r}$. The $n$-th channel sample of the time-varying radio AL is denoted by $h\!\left(n\right)$ with mean zero and variance $\upsilon_{h}$, while the channel fading of the quasi-static flat BL is denoted by $g$ which has the complex Gaussian distribution of mean zero and variance $\upsilon_{g}$. The transmit power of the MS and RRH are denoted by $P_{s}$ and $P_{r}$, respectively. The noise variance at the RRH and BBU pool is denoted by $\sigma_{n}^{2}$. It is assumed that the BBU pool acquires the knowledge of $\mathbf{t}_{s}$, $\mathbf{t}_{r}$, $\upsilon_{h}$, $\upsilon_{g}$, $P_{s}$, $P_{r}$, and $\sigma_{n}^{2}$.

During each transmission block, the MS transmits a signal $\mathbf{s}$ of
$N_{s}$ symbol length to the RRU, in which the $n$-th entry of $\mathbf{s}$ is given
by
\begin{align}
s\!\left(n\right)\!=\!\sqrt{1\!-\!\epsilon}b\!\left(n\right)\!+\!\sqrt{\epsilon} t_{s}\!\left(n\right),\quad 0\!\leq\!n\!\leq\!\left(N_{s}\!-\!1\right),
\end{align}
where $b\!\left(n\right)$ denotes the $n$-th entry of $\mathbf{b}$ with M-ary phase shift keying (MPSK) modulation constrained by $\mathcal{E}\!\big\{|b\!\left(n\right)|^{2}\!\big\}\!\!=\!\!P_{s}$, and
$t_{s}\!\left(n\right)$ represents the $n$-th entry of $\mathbf{t}_{s}$ with
$|t_{s}\!\left(n\right)\!|^{2}\!\!=\!\!P_{s}$ whose period is
denoted by $N_{p}$. The $\epsilon$ is within $0\!<\!\epsilon\!<\!1$.
Without loss of generality, we further assume that
$K\!=\!\frac{N_{s}}{N_{p}}$ is an integer\cite{7}. The $n$-th
observation at the RRH is written as
\begin{align}
x_{R}\!\left(n\right)\!=\!h\!\left(n\right)s\!\left(n\right)\!+\!w_{R}\!\left(n\right),\quad 0\!\leq\!n\!\leq\!\left(N_{s}\!-\!1\right),
\end{align}
where $w_{\!R}\!\left(n\right)$ is additive white Gaussian noise
(AWGN) at the RRH. Then the RRH scales the received
signal by
$\alpha\!\!=\!\!\sqrt{\frac{P_{r}}{\upsilon_{h}P_{s}\!+\!\sigma_{n}^{2}}}$, and inserts $\mathbf{t}_{r}$ prior to the received signal. The sequence $\mathbf{t}_{r}$ is of $N_{r}$ length and its $n$-th entry satisfies $|t_{r}\!\left(n\right)\!|^{2}\!\!=\!\!P_{r}$. The BBU pool receives two separate signals as
\begin{align}
\mathbf{x}_{s}\!&=\!\alpha g \cdot \mathrm{diag}\!\left(\mathbf{h}\right)\cdot\mathbf{s}\!+\!\alpha g \mathbf{w}_{R}\!+\!\mathbf{w}_{Ds},\label{EQ9}\\
\mathbf{x}_{r}\!&=\!g \mathbf{t}_{r}\!+\!\mathbf{w}_{Dr},
\end{align}
where
$\mathbf{w}_{\!R}\!\!\!=\!\!\!\left[w_{\!R}\!\left(1\right)\ldots w_{\!R}\!\left(\!N_{s}\right)\right]^{T}$,
$\mathbf{w}_{\!Ds}\!\!\!=\!\!\!\left[w_{\!Ds}\!\left(1\right)\ldots w_{\!Ds}\!\left(N_{s}\right)\right]^{T}$
and
$\mathbf{w}_{\!Dr}\!\!=\!\!\left[w_{Dr}\!\left(1\right),\ldots,w_{Dr}\!\left(N_{r}\right)\right]^{T}$
are AWGN vectors with each entry having variance $\sigma_{n}^{2}$.
In order to perform coherent reception and adopt cooperative
processing techniques at the BBU pool, the knowledge of
$h\!\left(n\right)$'s and $g$ should be obtained, which is explained
in the following section.

\begin{figure}[!h]
\center
  \includegraphics[width=0.5\textwidth]{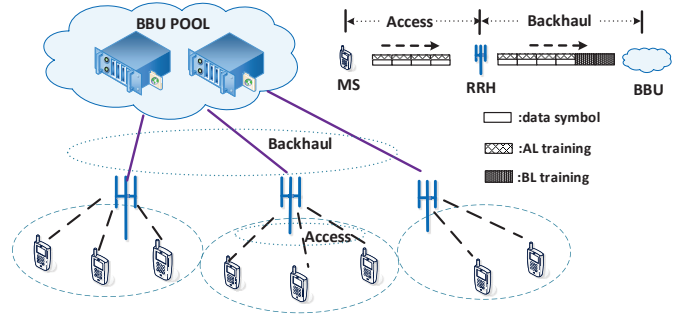}
 \caption{System model and the superimposed-segment training design.} \label{fig1}
\end{figure}

\section{CE-BEM Based MAP Channel Estimation Algorithm}
The CE-BEM for time-selective but frequency-flat fading channels in \cite{9}\cite{10} is chosen to model the time-varying AL as
\begin{align}\label{EQ8}
h\!\left(n\right)\!=\!\sum_{q\!=\!-Q}^{Q}\!\lambda_{q}e^{j\frac{2\pi q n}{N_{s}}},\quad 0\!\leq\!n\!\leq\!\left(N_{s}\!-\!1\right),
\end{align}
where $\lambda_{q}$'s are the BEM coefficients that remain invariant within one transmission block, and $\lambda_{q}$ has the complex Gaussian distribution with mean zero
and variance $\upsilon_{q}$ constrained to $\sum_{q\!=\!-Q}^{Q}\upsilon_{q}\!=\!\upsilon_{h}$. Substituting \eqref{EQ8} into \eqref{EQ9}, $\mathbf{x}_{r}$ can be rewritten as
\begin{align}
\mathbf{x}_{s}\!=\!\alpha g \mathbf{D} \left[\mathbf{I}_{2Q\!+\!1}\!\otimes\!\mathbf{s}\right]\boldsymbol{\lambda}\!+\!\alpha g \mathbf{w}_{R}\!+\!\mathbf{w}_{Ds},
\end{align}
where $\mathbf{D}\!\!=\!\!\big[\mathbf{D}_{-Q},\ldots,\mathbf{D}_{Q}\big]$ is an $N_{s}\!\times\!\left(2Q\!+\!1\right)N_{s}$ dimensional matrix with $\mathbf{D}_{q}\!\!=\!\!\mathrm{diag}\!\left(1,e^{j\frac{2\pi q}{N_{s}}},\ldots,e^{j\frac{2\pi q \left(N_{s}\!-\!1\right)}{N_{s}}}\right)$. By defining $\mathbf{J}\!=\!\frac{1}{K}\!\mathbf{1}_{K}^{T}\!\otimes\!\mathbf{I}_{N_{p}}$ of $N_{p}\!\times\!N_{s}$ dimension, we can obtain
\begin{subnumcases}
{\mathbf{J}\mathbf{D}_{q}\mathbf{t}_{s}\!=\!}
\mathbf{\tilde{t}}_{s}, & $q = 0$,\\
0, & $q\neq 0$,
\end{subnumcases}
where $\mathbf{\tilde{t}}_{s}$ is an $N_{p} \times 1$ dimensional
vector. Left multiplying $\mathbf{x}_{s}$ by
$\left(\mathbf{I}_{Q\!+\!1}\!\otimes\!\mathbf{J}\right)\!\mathbf{D}^{H}$ yields
\begin{align}
\mathbf{r}\!=\!\left(\mathbf{I}_{Q\!+\!1}\!\otimes\!\mathbf{J}\right)\!\mathbf{D}^{H}\mathbf{x}_{s},
\end{align}
whose
$\left[\left(q\!+\!Q\right)\!N_{p}\!+\!1\right]$-th to
$\left[\left(q\!+\!Q\!+\!1\right)N_{p}\right]$-th entries denoted by
$\mathbf{r}_{q}$ can be expressed as
\begin{align}
\mathbf{r}_{q}\!=\!\alpha g \mathbf{\tilde{t}}_{s} \lambda_{q}\!+\!\underbrace{\alpha g\!\!\!\sum_{l=-Q}^{Q}\!\mathbf{J}\mathbf{D}_{l\!-\!q}\mathbf{b}\lambda_{l}\!+\!\mathbf{J}\mathbf{D}_{q}^{H}\left(
\alpha g \mathbf{w}_{R}\!\!+\!\!\mathbf{w}_{Ds}\right)}_{\mathbf{w}_{q}}.
\end{align}
It is testified that
\begin{align}
\mathcal{E}\!\big\{\mathbf{w}_{q1}\mathbf{w}_{q1}^{H}\big\}\!=\!\upsilon_{n}\mathbf{I}_{N_{p}},\quad
\mathcal{E}\!\big\{\mathbf{w}_{q1}\mathbf{w}_{q2}^{H}\big\}\!=\!\mathbf{0}_{N_{p}}
\end{align}
for $q_{1}\neq q_{2}$ with
$\upsilon_{n}\!\!=\!\!\frac{\alpha^{2}\upsilon_{g}\!\upsilon_{h}\!\left(1\!-
\!\epsilon\right)\!P_{s}\!+\!\left(\alpha^{2}\upsilon_{g}\!+\!1\right)\!\sigma_{n}^{2}}{K}$.
Although $\mathbf{w}_{q}$ is not a Gaussian vector, it is effective to use the Gaussian distribution to model the noise behavior for estimation problems \cite{14}. Thus, we choose the following function to be the nominal likelihood function of $\mathbf{r}$ as
\begin{align}
p\!\left(\mathbf{r}|\boldsymbol{\lambda},\!g\right)\!=\!\!\prod_{q\!=\!-Q}^{Q}\!\left(\!\frac{1}{\pi \upsilon_{n}}\!\right)^{\!\!N_{p}}\!\!\!e^{\!-\!\frac{\|\mathbf{r}_{q}\!-\!\alpha g \lambda_{q}\mathbf{\tilde{t}}_{s}\|^{2}}{\upsilon_{n}}}.
\end{align}

\subsection{Estimation for $\lambda_{q}$'s and $g$}
Defining $\boldsymbol{\theta}\!=\!\left[\lambda_{-Q},\ldots,\lambda_{Q},g\right]^{T}$, the MAP estimation for $\boldsymbol{\theta}$ gives
\begin{align}\label{EQ1}
&\hat{\boldsymbol{\theta}}\!=\!\arg \max_{\boldsymbol{\theta}}\!\Big\{p\!\left(\mathbf{r}|\boldsymbol{\lambda},\!g \right)\!
p\!\left(\mathbf{x}_{r}|g \right)\!
p\!\left(\boldsymbol{\lambda}\right) p\!\left(g\right)\Big\}\\
\!&=\!\arg \min_{\boldsymbol{\theta}}\!\Bigg\{\!\underbrace{\sum_{q\!=\!-Q}^{Q}\!\!\!
\bigg\{\!\frac{\|\mathbf{r}_{q}\!\!-\!\!\alpha g \lambda_{q}\mathbf{\tilde{t}}_{s}\|^{2}}{\upsilon_{n}}\!\!+\!\!{\frac{|\lambda_{q}|^{2}}{\upsilon_{q}}}\!\!
\bigg\}\!\!+\!\!\frac{\|\mathbf{x}_{r}
\!\!-\!\!g\mathbf{t}_{r}\|^{2}}{\sigma_{n}^{2}}\!\!+\!\!\frac{|g|^{2}}{\upsilon_{g}}}_{
\mathcal{L}\left(\boldsymbol{\theta}\right)}\!\Bigg\},\nonumber
\end{align}
where $p\!\left(\mathbf{x}_{r}|g \right)$, $\!
p\!\left(\boldsymbol{\lambda}\right)$ and $p\!\left(g\right)$ are Gaussian distribution functions.
With a given $g$, the estimate of $\lambda_{q}$ can be obtained as
\begin{align}\label{EQ2}
\hat{\lambda}_{q}\!=\!\frac{\alpha g^{H}\tilde{\mathbf{t}}_{s}^{H}\mathbf{r}_{q}}{\alpha^{2}|g|^{2}\|\mathbf{\tilde{t}}_{s}\|^{2}
\!+\!\frac{\upsilon_{n}}{\upsilon_{q}}},\!-\!Q\!\leq\!q\!\leq\!Q.
\end{align}
Substituting \eqref{EQ2} into \eqref{EQ1}, the estimate of $g$ can be obtained from
\begin{align}
\hat{g}\!=\!\arg \min_{g}\!\!\Bigg\{\!\!\underbrace{\frac{\|\mathbf{x}_{r}
\!\!-\!\!g\mathbf{t}_{r}\|^{2}}{\sigma_{n}^{2}}\!\!+\!\!\frac{|g|^{2}}{\upsilon_{g}}\!\!-\!\!
\frac{\alpha^{2}|g|^{2}\!\!\sum_{q\!=\!-Q}^{Q}\!|\mathbf{t}_{s}^{\!H}\mathbf{r}_{q}|^{2}}{\alpha^{2}\upsilon_{n}
\!\|\mathbf{\tilde{t}}_{s}\|^{\!2}|g|^{2}\!\!+\!\!\frac{\upsilon_{n}^{2}}{\upsilon_{q}}}}_{\mathcal{L}
\!\left(g\right)}\!\!\Bigg\}.
\end{align}
Note that, only the first term in $\mathcal{L}\!\left(g\right)$ relates to the phase of $g$ denoted by $\angle{g}$, and thus the estimate of $\angle{g}$ can be directly estimated from minimizing $\|\mathbf{x}_{r}\!-\!g\mathbf{t}_{r}\|^{2}$ as
\begin{align}\label{EQ3}
\widehat{\angle{g}}\!=\!\arg \min\!\Big\{\|\mathbf{x}_{r}
\!\!-\!\!g\mathbf{t}_{r}\|^{2}\!\Big\}\!=\!\angle{\left(\mathbf{t}_{r}^{H}\mathbf{x}_{r}\right)}.
\end{align}
The estimate of $|g|$ must be either a local minima of $\mathcal{L}\!\left(|g|\big|\angle{g}\right)$ or at the boundary $|g|\!=\!0$, which can be obtained from solving $\frac{\partial \mathcal{L}\!\left(|g|\big|\angle{g}\right)}{\partial |g|}\!=\!0$. Unfortunately, a closed-form expression for $|\hat{g}|$ is hard to derive since $\frac{\partial \mathcal{L}\!\left(|g|\big|\angle{g}\right)}{\partial |g|}$ is an $m$($\geq4$)th-order polynomial of $|g|$, and thus numerical methods such as one dimensional search are needed to compute the value of $\widehat{|g|}$. To reduce the complexity of such approaches, an iterative approach is developed whereby $\hat{g}$ is initialized from
\begin{align}\label{EQ4}
\hat{g}\!&=\!\arg \max_{g}\Big\{\!p\!\left(\mathbf{x}_{r}|g\right)p\!\left(g\right)\!\Big\}\!=\!\frac{\mathbf{t}_{r}^{H}\mathbf{x}_{r}}{\|\mathbf{t}_{r}\|^{2}\!+\!\frac{\sigma_{n}^{2}}{\upsilon_{g}}}.
\end{align}
With $\hat{g}$ obtained, $\hat{\lambda}_{q}$'s can be estimated
according to \eqref{EQ2} with $\hat{g}$ in place of $g$. Then,
$\hat{g}$ can be further updated by substituting
$\hat{\lambda}_{q}$'s into \eqref{EQ3} as
\begin{align}\label{EQ5}
\hat{g}\!=\!\frac{\sum_{q\!=\!-Q}^{Q}\!\!\frac{\alpha \hat{\lambda}_{q}^{H}\mathbf{\tilde{t}}_{s}^{H}\mathbf{r}_{q}}{\upsilon_{n}}\!+\!\frac{\mathbf{t}_{r}^{H}\mathbf{x}_{r}
}{\sigma_{n}^{2}}}{\sum_{q\!=\!-Q}^{Q}\!\!\frac{\alpha^{2}|\hat{\lambda}_{q}|^{2}\|\mathbf{\tilde{t}}_{s}\|^{2}}{\upsilon_{n}}\!+\!
\frac{\|\mathbf{t}_{r}\|^{2}}{\sigma_{n}^{2}}\!+\!\frac{1}{\upsilon_{g}}}.
\end{align}

\subsection{Restoration for Channel Samples $h\!\left(n\right)$'s}

With $\hat{\lambda}_{q}$'s obtained, $\hat{h}\!\left(n\right)$ is restored as
\vspace{-0.05in}
\begin{align}\label{EQ7}
\hat{h}\!\left(n\right)\!=\!\sum_{q\!=\!-\!Q}^{Q}\!\Big\{\eta_{q}\cdot
\hat{\lambda}_{q}e^{j\frac{2\pi q n}{N_{s}}}\Big\}, \quad 0\!\leq\!n\!\leq\! \left(N_{s}\!-1\right),
\end{align}
\vspace{-0.05in}
where $\eta_{q}$'s are real factors. The vector $\boldsymbol{\eta}$ that maximizes the average effective SNR (AESNR) \cite{11} denoted by $\boldsymbol{\eta}^{*}$ is obtained from
\begin{align}
\boldsymbol{\eta}^{*}\!\!=\!\! \arg \!\max_{\boldsymbol{\eta}} \! \underbrace{\frac{\mathcal{E}\!\bigg\{\mathcal{E}\!\Big\{
\big|\hat{g}\hat{h}\!\left(n\right)\big|^{2}\Big|\boldsymbol{\lambda},\!g\Big\}\bigg\}\!
\left(1\!-\!\epsilon\right)}{\mathcal{E}\!\bigg\{\!\mathcal{E}\!\!\Big\{\!\!
\big|\hat{g}\hat{h}\!\left(n\right)\!\!-\!\!gh\!\left(n\right)\!\!\big|^{2}\!\Big|
\boldsymbol{\lambda},g\Big\}\!\!+\!\!\left(\!|g|^{\!2}\!\!+\!\!\frac{1}{\alpha^{\!2}}
\right)\!\!\frac{\sigma_{n}^{2}}{P_{s}}\bigg\}}\!}_{\bar{\gamma}}.
\end{align}
Define $\phi_{q,n}\!\!=\!\!e^{\!j\frac{2\pi q n}{N_{s}}}$, and
denote $\Delta{g}$ and $\Delta{\lambda_{q}}$ to be the
estimation error of $g$ and $\lambda_{q}$, respectively. It is obtained that
\begin{align}\label{EQN1}
&\mathcal{E}\!\Big\{\!
\big|\hat{g}\hat{h}\!\left(n\right)\!\big|^{2}\!\Big|\boldsymbol{\lambda},\!g\!\Big\}\!\!=\!\!|g|^{2}\!\!\!
\sum_{q\!=\!-Q}^{Q}\!\!\!\eta_{q}^{2}\!\underbrace{\mathcal{E}\!
\Big\{\!|\Delta{\!\lambda_{q}}|^{\!2}\!\Big\}}
_{\delta_{\lambda_{q}}}\!+\!\bigg|\!\underbrace{\sum_{q\!=\!-Q}^{Q}\!\!\!\eta_{q}\!\lambda_{q}\phi_{q,n}}
_{\tilde{h}\!\left(n\right)}\!
\bigg|^{2}\!\!\underbrace{\!\mathcal{E}\!\Big\{\!|\Delta{g}|^{\!2}\!\Big\}}_{\delta_{g}}
\!+\!|g|^{2}\!\bigg|\!\sum_{q\!=\!-Q}^{Q}\!\!\!\eta_{q}\!\lambda_{q}
\phi_{q,n}\!\bigg|^{2}\!+\!\mathcal{E}\!
\Bigg\{\bigg|\Delta{g}\!\sum_{q\!=\!-Q}^{Q}\!\eta_{q}\Delta{\!\lambda_{q}}\phi_{q,n}\bigg|^{2}\Bigg\}.
\end{align}
Note that, $\delta_{\lambda_{q}},\delta_{g}\!\sim\!\mathcal{O}\!\left(\frac{\sigma_{n}^{2}}{P_{s}}\right)$ while the last term in \eqref{EQN1} has the order of $\mathcal{O}\!\!\left[\!\left(\frac{\sigma_{n}^{2}}{P_{s}}\right)^{\!2}\!\right]$, thus we remove the last term for high SNR approximation, i.e.,
\begin{align}
\mathcal{E}\!\Big\{\!
\big|\hat{g}\hat{h}\!\left(n\right)\!\big|^{2}\!\Big|\boldsymbol{\lambda},\!g\!\Big\}\!\!=\!|g|^{2}\!\!\!
\sum_{q\!=\!-Q}^{Q}\!\!\!\eta_{q}^{2}\delta_{\lambda_{q}}\!\!+\!\big|\tilde{h}\!\left(n\right)\!\!\big|^{2}\delta_{g}
\!\!+\!\!|g|^{2}\!\big|\tilde{h}\!\left(n\right)\!\big|^{2}.
\end{align}
Similarly,
\begin{align}
\mathcal{E}\!\Big\{\!
\big|\hat{g}\hat{h}\!\left(n\right)\!\!-\!\!gh\!\left(n\right)\!\big|^{2}\!\Big|\boldsymbol{\lambda},\!g\!\Big\}
\!\!\approx&\!|g|^{2}\!\!\!
\sum_{q\!=\!-Q}^{Q}\!\!\!\eta_{q}^{2}\delta_{\lambda_{q}}\!\!+\!\big|\tilde{h}\!\left(n\right)\!\big|^{2}\delta_{g}
\!\!+\!\!|g|^{2}\!\bigg|\!\sum_{q\!=\!-Q}^{Q}\!\left(\eta_{q}\!\!-\!\!1\right)\!\lambda_{q}
\phi_{n,q}\!\bigg|^{2}.
\end{align}
Due to the non-linearity of the MAP estimation, it is hard to derive the corresponding MSE expressions in closed forms. Moreover, it is known that the MAP estimation MSEs converge to the Cram\'{e}r-Rao bound (CRB) \cite{12} when the training length is sufficiently large, and thus it is effective to use the CRBs in the computation of the AESNR as
\begin{align}
\delta_{\!\lambda_{q}}\!\!=\!\!\mathrm{CRB}_{\!\lambda_{q}}\!\!=\!\!\frac{\upsilon_{n}}{\alpha^{2}|g|^{2}\|\mathbf{\tilde{t}}_{s}\|^{2}},
\delta_{g}\!=\!\mathrm{CRB}_{g}\!\!=\!\!\frac{1}{\frac{\alpha^{\!2}\|\boldsymbol{\lambda}\|^{\!2}\|\mathbf{\tilde{t}}_{s}\!\|^{\!2}}{\upsilon_{n}}
\!\!+\!\!\frac{\|\mathbf{t}_{r}\|^{2}}{\sigma_{n}^{2}}}.
\end{align}
Substituting the above expressions for $\delta_{\lambda_{q}}$ and $\delta_{g}$ into $\bar{\gamma}$ yields
\begin{align}\label{EQ16}
\bar{\gamma}\!=\!\frac{\boldsymbol{\eta}^{T} \mathbf{\Xi} \boldsymbol{\eta}\left(1\!-\!\epsilon\right)}{\boldsymbol{\eta}^{T}\mathbf{\Xi}\boldsymbol{\eta} -2\mathbf{1}^{T}_{2Q\!+\!1}\mathbf{\Upsilon}\boldsymbol{\eta}+C},
\end{align}
where
\begin{align}
&\mathbf{\Upsilon}\!\!=\!\!\frac{\alpha^{2}\upsilon_{g}\|\mathbf{\tilde{t}}_{s}\|^{2}}
{\upsilon_{n}}\mathbf{R}_{\lambda}^{2}\!\!+\!\!\frac{\alpha^{2}\upsilon_{g}\!\|\mathbf{\tilde{t}}_{s}\|^{2}}
{\upsilon_{n}}\mathrm{tr}\!\big\{\!\mathbf{R}_{\lambda}\!\big\}\mathbf{R}_{\lambda}
\!\!+\!\!\frac{\upsilon_{g}\!\|\mathbf{t}_{r}\|^{2}}{\sigma_{n}^{2}}\mathbf{R}_{\lambda},\\
&\mathbf{\Xi}\!\!=\!\!\mathbf{\Upsilon}+\mathbf{R}_{\lambda}\!+\!\left(\upsilon_{h}\!
+\!\frac{\upsilon_{n}\|\mathbf{t}_{r}\|^{2}}{\sigma_{n}^{2}\alpha^{2}\|\mathbf{\tilde{t}}_{s}\|^{2}}\right)
\mathbf{I}_{2Q\!+\!1},\\
&C\!\!=\!\!\mathbf{1}^{T}_{2Q\!+\!1}\mathbf{\Upsilon}\mathbf{1}_{2Q\!+\!1}\!\!+\!\!\left(\upsilon_{g}
\!\!+\!\!\frac{1}{\alpha^{2}}\!\right)\!\!\left[\frac{\upsilon_{h}\alpha^{2}\|\mathbf{\tilde{t}}_{s}\|^{2}}
{\upsilon_{n}}\!+\!\!\frac{\|\mathbf{t}_{r}\|^{2}}{\sigma_{n}^{2}}\right]\!\frac{\sigma_{n}^{2}}{P_{s}},
\end{align}
and $\mathbf{R}_{\lambda}\!\!=\!\!\mathrm{diag}\!\left(\upsilon_{\!-\!Q}\ldots\upsilon_{Q}\right)$. On setting $\varphi^{*}\!\!=\!\!\left(\boldsymbol{\eta}^{*}\right)^{T}\mathbf{\Xi}\boldsymbol{\eta}^{*}$,
the optimization for $\boldsymbol{\eta}$ transforms to
\begin{align}
&\max_{\boldsymbol{\eta}} \quad 2\mathbf{1}^{T}_{2Q\!+\!1}\mathbf{\Upsilon}\boldsymbol{\eta},\label{EQ_max}\\
&\mathrm{s.t.}\quad\boldsymbol{\eta}^{T}\mathbf{\Xi}\boldsymbol{\eta}\!=\!\varphi^{*}.\label{EQ_max_con}
\end{align}

Clearly, the optimization problem described in \eqref{EQ_max}
constrained by \eqref{EQ_max_con} is concave, and
$\boldsymbol{\eta}^{*}$ can be directly obtained from the Lagrange
dual function as
\begin{align}\label{EQ10}
\boldsymbol{\eta}^{*}\!=\!\frac{\sqrt{\varphi^{*}}\mathbf{\Xi}^{-1}\mathbf{\Upsilon}\mathbf{1}_{2Q\!+\!1}}{\sqrt{\mathbf{1}_{2Q\!+\!1}^{T}\mathbf{\Upsilon}
\mathbf{\Xi}^{-1}\mathbf{\Upsilon}\mathbf{1}_{2Q\!+\!1}}}.
\end{align}
Substituting \eqref{EQ10} back to \eqref{EQ16}, the optimization
problem becomes
\begin{align}
&\max_{\varphi}\quad\frac{\varphi}{\varphi-2\sqrt{\mathbf{1}_{2Q\!+\!1}^{T}\mathbf{\Upsilon}
\mathbf{\Xi}^{-1}\mathbf{\Upsilon}\mathbf{1}_{2Q\!+\!1}}\sqrt{\varphi}+C},\\
&\mathrm{s.t.}\quad \varphi > 0,
\end{align}
whose solution is $
\varphi^{*}=\frac{C^{2}}{\mathbf{1}_{2Q\!+\!1}^{T}\mathbf{\Upsilon}
\mathbf{\Xi}^{-1}\mathbf{\Upsilon}\mathbf{1}_{2Q\!+\!1}}$ leading to
\begin{align}\label{EQ6}
\boldsymbol{\eta}^{*}=\frac{C\mathbf{\Xi}^{-1}\mathbf{\Upsilon}\mathbf{1}_{2Q\!+\!1}}{\mathbf{1}_{2Q\!+\!1}^{T}\mathbf{\Upsilon}
\mathbf{\Xi}^{-1}\mathbf{\Upsilon}\mathbf{1}_{2Q\!+\!1}}.
\end{align}

Combining the estimation for $\boldsymbol{\theta}$ and the restoration for $h\!\left(n\right)$'s, the proposed channel estimation algorithm is summarized in Table I. Moreover, the following proposition is given to demonstrate the effectiveness of the proposed algorithm:
\begin{proposition}
The iterative channel estimation algorithm is convergent, and it achieves lower MSE than that of the maximum likelihood (ML) method.
\end{proposition}
\begin{proof}
Each iteration consists of $\left(2Q\!+\!2\right)$
steps. Denote the $i$th entry of $\boldsymbol{\theta}$ by $\theta_{i}$,
and the updated estimate of $\theta_{i}$ denoted by
$\hat{\theta}_{i}^{\mathrm{new}}$ satisfies
$\mathcal{L}\!\left(\hat{\theta}_{i}^{\mathrm{new}}\right)\!>\!\mathcal{L}
\!\left(\hat{\theta}_{i}\right)$. This indicates that $\mathcal{L}\!\left(\boldsymbol{\theta}\right)$ strictly increases after each step as well as after one round of
iteration. It is known that $\mathcal{L}
\!\left(\boldsymbol{\theta}\right)$ is continuous with respect to $\theta_{i}$ and $\mathcal{L}
\!\left(\boldsymbol{\theta}\right)\!<\!+\infty$. Thus, it is concluded that the iterative algorithm is convergent.

From \eqref{EQ2}, the MAP estimate of $g$ with a given $\boldsymbol{\lambda}$ is
\begin{align}
\hat{g}^{\mathrm{MAP}}\!\!\!\!=\!\!\frac{\frac{\alpha^2\|\boldsymbol{\lambda}\|^{2}\|
\mathbf{\tilde{t}}_{s}\|^{2}}{\upsilon_{n}}\!+\!\frac{\|\mathbf{t}_{r}\|^{2}}{\sigma_{n}^{2}}}{\frac{\alpha^2\|\boldsymbol{\lambda}\|^{2}\|
\mathbf{\tilde{t}}_{s}\|^{2}}{\upsilon_{n}}\!\!+\!\!\frac{\|\mathbf{t}_{r}\|^{2}}{\sigma_{n}^{2}}\!\!+\!\!\frac{1}{\upsilon_{g}}}
g\!+\!\frac{\sum_{q\!=\!-Q}^{Q}\!\!\frac{\alpha \lambda_{q}^{H}\mathbf{\tilde{t}}_{s}^{H}\mathbf{w}_{q}}{\upsilon_{n}}\!\!+\!\!\frac{\mathbf{t}_{r}^{H}\mathbf{w}_{Dr}
}{\sigma_{n}^{2}}}{\frac{\alpha^2\|\boldsymbol{\lambda}\|^{2}\|
\mathbf{\tilde{t}}_{s}\|^{2}}{\upsilon_{n}}\!\!+\!\!\frac{\|\mathbf{t}_{r}\|^{2}}{\sigma_{n}^{2}}\!\!+\!\!\frac{1}{\upsilon_{g}}},\nonumber
\end{align}
whose MSE is calculated as
\begin{align}\label{EQ11}
\delta_{g}^{\mathrm{MAP}}\!\!\!\!&=\!\mathcal{E}\!\Big\{\!\big\|\hat{g}^{\mathrm{MAP}}\!\!\!-\!g\big\|^{\!2}
\!\Big\}\!\!\nonumber\\&=\!\!\frac{1}{\frac{\alpha^{\!2}\|\boldsymbol{\lambda}\|^{\!2}\|
\mathbf{\tilde{t}}_{s}\|^{\!2}}{\upsilon_{n}}\!\!+\!\!\frac{\|\mathbf{t}_{r}\|^{\!2}}{\sigma_{n}^{2}}}
\frac{\upsilon_{g}}{\upsilon_{g}\!\!+\!\!\frac{1}{\frac{\alpha^{\!2}\|\boldsymbol{\lambda}\|^{\!2}\|
\mathbf{\tilde{t}}_{s}\|^{\!2}}{\upsilon_{n}}\!+\!\frac{\|\mathbf{t}_{r}\|^{\!2}}{\sigma_{n}^{2}}}}.
\end{align}
The ML estimate of $g$ gives
\begin{align}
\hat{g}^{\mathrm{ML}}\!\!=\!g\!+\!\frac{\sum_{q\!=\!-Q}^{Q}\!\!\frac{\alpha \lambda_{q}^{H}\mathbf{\tilde{t}}_{s}^{H}\mathbf{w}_{q}}{\upsilon_{n}}\!\!+\!\!\frac{\mathbf{t}_{r}^{H}\mathbf{w}_{Dr}
}{\sigma_{n}^{2}}}{\frac{\alpha^2\|\boldsymbol{\lambda}\|^{2}\|
\mathbf{\tilde{t}}_{s}\|^{2}}{\upsilon_{n}}\!\!+\!\!\frac{\|\mathbf{t}_{r}\|^{2}}{\sigma_{n}^{2}}},\nonumber
\end{align}
whose MSE is calculated as
\begin{align}\label{EQ12}
\delta_{g}^{\mathrm{ML}}\!&=\!\mathcal{E}\!\Big\{\!\big\|\hat{g}^{\mathrm{ML}}\!\!\!-\!g\big\|^{\!2}
\!\Big\}\!=\!\frac{1}{\frac{\alpha^2\|\boldsymbol{\lambda}\|^{2}\|
\mathbf{\tilde{t}}_{s}\|^{2}}{\upsilon_{n}}\!+\!\frac{\|\mathbf{t}_{r}\|^{2}}{\sigma_{n}^{2}}}.
\end{align}
Clearly, $\delta_{g}^{\mathrm{MAP}}\!\!<\!\delta_{g}^{\mathrm{ML}}$ always holds, and it can be also testified that $\delta_{\lambda_{q}}^{\mathrm{MAP}}\!\!<\!\!\delta_{\lambda_{q}}^{\mathrm{ML}}$ is satisfied similarly.
\end{proof}

\emph{Remark}: According to \eqref{EQ11} and \eqref{EQ12}, we see that $\delta_{g}^{\mathrm{MAP}}\!\!<\!\!\mathrm{CRB}_{g}$ and $\delta_{g}^{\mathrm{ML}}\!\!=\!\!\mathrm{CRB}_{g}$. This is because the proposed MAP estimation algorithm is biased since $\mathcal{E}\!\{\hat{g}^{\mathrm{MAP}}\!\}\!\!\neq\!\!g$.

\begin{table}[!t]
\center
\caption{Iterative Channel Estimation Algorithm}
\begin{tabular}{p{3.4in}@{}}
\hline
\begin{itemize}
\item \textbf{Initialize} $\hat{g}$ in accordance with \eqref{EQ4}; i\_Index$=$i\_Times($>0$)
\item \textbf{Repeat}
\begin{itemize}
\item[-] For each $q$, update $\hat{\lambda}_{q}$ by substituting $\hat{g}$ into \eqref{EQ2}.
\item[-] Update $\hat{g}$ by substituting $\hat{\lambda}_{q}$'s into \eqref{EQ5}.
\item[-] i\_Index$=$i\_Index$-1$
\end{itemize}
\item \textbf{Until} i\_Index$\leq 0$ is satisfied.
\item \textbf{Calculate} the optimal $\boldsymbol{\eta}^{*}$ according to \eqref{EQ6}.
\item \textbf{Restore} $\hat{h}\!\left(n\right)$'s according to \eqref{EQ7} by using $\hat{\lambda}_{q}$'s and $\boldsymbol{\eta}^{*}$.
\item \textbf{Return} $\hat{h}\!\left(n\right)$'s and $\hat{g}$.
\end{itemize}
\\ \hline
\end{tabular}
\end{table}

\section{Numerical Results}

Numerical results are provided to evaluate the
performance of the CE-BEM based MAP (C-MAP) channel
estimation algorithm. The AL channel
$\big\{\!h\!\left(n\right)\!\!\big\}$ and BL channel $g$ are generated from the spatial channel model (CSM) in \emph{3GPP TR 25.996} \cite{13}. The parameters are set as $N_{s}\!=\!800$, $N_{r}\!=\!4$ and $\mathrm{i\_Times}\!\!\!=\!\!\!10$. We assume binary-phase-shift-keying
(BPSK) modulation for $\big\{\!b\!\left(n\right)\!\!\big\}$, while
$\mathbf{\tilde{t}}_{s}$ is selected as the $2$nd column of the $N_{p}\!\times\!N_{p}$ discrete Fourier
transform (DFT) matrix and $\mathbf{t}_{r}$ is as the $3$rd column of the $N_{r}\!\times\!N_{r}$ DFT matrix. The transmit power $P_{s}$ and
$P_{r}$ are set to be equal, and the noise variance is set to be of
unit value; thus the SNR is equal to $P_{s}$.

In Fig. \ref{fig2}, the average MSEs of both AL and BL channels for
C-MAP estimation are compared with that for ML estimation. It is
observed that the proposed C-MAP estimation outperforms the
traditional ML method since C-MAP achieves lower MSEs of both
AL and BL channels than ML. Moreover, it is seen that the MSE of
the AL channel for the ML method is not convergent. This is because
random generation of $g$ would result in singularity, e.g.,
$g\!\!\rightarrow\!\!0$, leading to
$\Delta\!\lambda_{q}\!\!\rightarrow\!\!\infty$ for the ML method, while
the proposed C-MAP algorithm is robust at the singularity. In
Fig.\ref{fig3}, we evaluate the AESNR performance for the optimal
weighted approach (OWA) to channel restoration. It is seen that OWA
obtains higher AESNRs than the baseline (restoring according to
CE-BEM), especially in the low SNR region, and it draws near to the
baseline as SNR increases.

\begin{figure}[!t]
\center
  \includegraphics[width=0.5\textwidth]{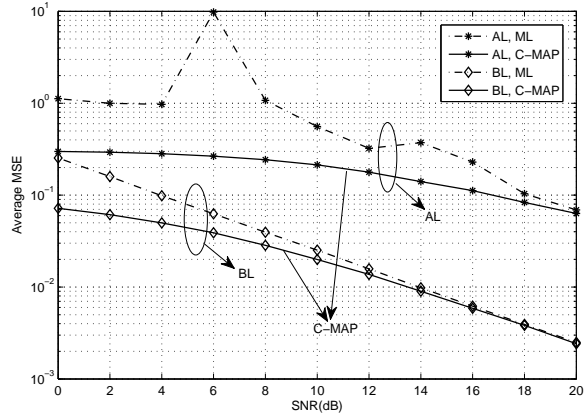}
 \caption{Average MSE versus SNR for different estimation methods.} \label{fig2}
\end{figure}

\begin{figure}[!t]
\center
  \includegraphics[width=0.5\textwidth]{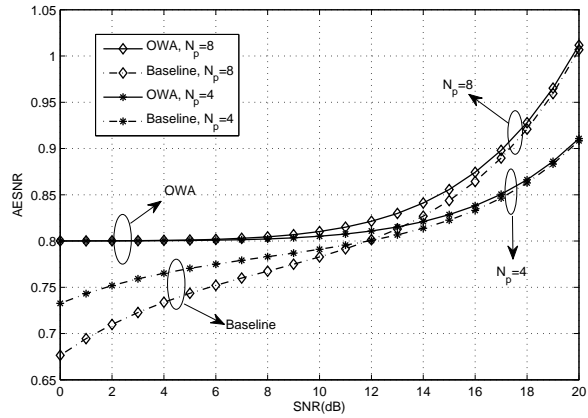}
 \caption{AESNR versus SNR for different channel restoring methods.} \label{fig3}
\end{figure}

\section{Conclusions}

A superimposed-segment training design has been
proposed to decrease the training overhead and enhance the channel
estimation accuracy in uplink C-RANs. Based on the training design,
a CE-BEM based MAP channel estimation algorithm has been developed,
where the BEM coefficients of the time-varying radio AL and the
channel fading of the quasi-static wireless BL are first obtained,
then the time-domain channel samples of AL are restored from
maximizing AESNR. Simulation results have demonstrated that the
proposed algorithm declines the estimation MSE
and increases the AESNR.

\end{document}